\documentclass[journal]{IEEEtran}

\usepackage{epstopdf}
\usepackage[draft]{hyperref} 
\usepackage{setspace}
\usepackage{amsmath}
\usepackage{amssymb}
\usepackage{stfloats}
\usepackage{amsthm}
\usepackage{multirow}
\usepackage{amsfonts}
\usepackage{color}
\usepackage{graphicx}
\usepackage{subfigure}  
\usepackage[]{algorithmicx}
\usepackage{algpseudocode,algorithm}
\usepackage{cite}
\usepackage{mathtools}
\usepackage{stmaryrd}
\usepackage[mathscr]{euscript}
\usepackage{array}
\usepackage{mathrsfs}

\usepackage{soul}

\newtheorem{theorem}{Theorem}

\usepackage{authblk}

\author{   Hila Naaman, {\it Student Member, IEEE}, Daniel Bilik, Shlomi Savariego, Moshe Namer, Yonina C. Eldar, {\it Fellow, IEEE}
	\thanks{\scriptsize The authors are with the Faculty of Mathematics
    and Computer Science, Weizmann Institute of Science, Rehovot 76100, Israel
    (e-mail: hilanaaman10@gmail.com; daniel.bilik@weizmann.ac.il;  shlomi.savariego@weizmann.ac.il; moshe.namer@weizmann.ac.il; yonina@weizmann.ac.il).}
	\thanks{\scriptsize {
    This research was partially supported by the Israeli Council for Higher Education (CHE) via the Weizmann Data Science Research Center, by a research grant from the Estate of Tully and Michele Plesser, by the European Union’s Horizon 2020 research and innovation program under grant No. 646804-ERC-COG-BNYQ, by the Israel Science Foundation under grant no. 0100101, and by the Tom and Mary Beck Center for Renewable Energy as part of the Institute for Environmental Sustainability (IES) at the Weizmann Institute of Science.}\vspace{-0.3cm}}
}

\begin{document}

\title{ECG-TEM: Time-based sub-Nyquist sampling for ECG signal reconstruction and Hardware Prototype}
\maketitle

\begin{abstract}
Portable heart rate monitoring (HRM) systems based on electrocardiograms (ECGs) have become increasingly crucial for preventing lifestyle diseases. For such portable systems, minimizing power consumption and sampling rate is critical due to the substantial data generated during long-term ECG monitoring. The variable pulse-width finite rate of innovation (VPW-FRI) framework provides an effective solution for low-rate sampling and compression of ECG signals.
We develop a time-based sub-Nyquist sampling and reconstruction method for ECG signals specifically designed for HRM applications. Our approach harnesses the integrate-and-fire time-encoding machine (IF-TEM) as a power-efficient, time-based, asynchronous sampler, generating a sequence of time instants without the need for a global clock.
The ECG signal is represented as a linear combination of VPW-FRI pulses, which is then subjected to pre-filtering before being sampled by the IF-TEM sampler. A compactly supported robust filter with a frequency-domain alias cancellation condition is used to combat the effects of noise. Our reconstruction process involves consecutive partial summations of discrete representations of the input signal derived from the series of time encodings, further enhancing the accuracy of the reconstructed ECG signals. Additionally, we introduce an IF-TEM sampling hardware system for ECG signals, implemented using an analog filter device. The firing rate is 42-80Hz, equivalent to approximately 0.025-0.05 of the Nyquist rate. Our hardware validation bridges the gap between theory and practice and demonstrates the robust performance and practical applicability of our approach in accurately monitoring heart rates and reconstructing ECG signals. 
\end{abstract}

%
\begin{IEEEkeywords}
Analog-to-digital converter (ADC), Time-encoding machine (TEM), ECG signal, time-based sampling, integrate-and-fire, heart rate monitoring, sampling theory.
\end{IEEEkeywords}
\IEEEpeerreviewmaketitle

\section{Introduction}
\begin{figure}[t!]
		\centering
		\includegraphics[width = 0.5\textwidth]{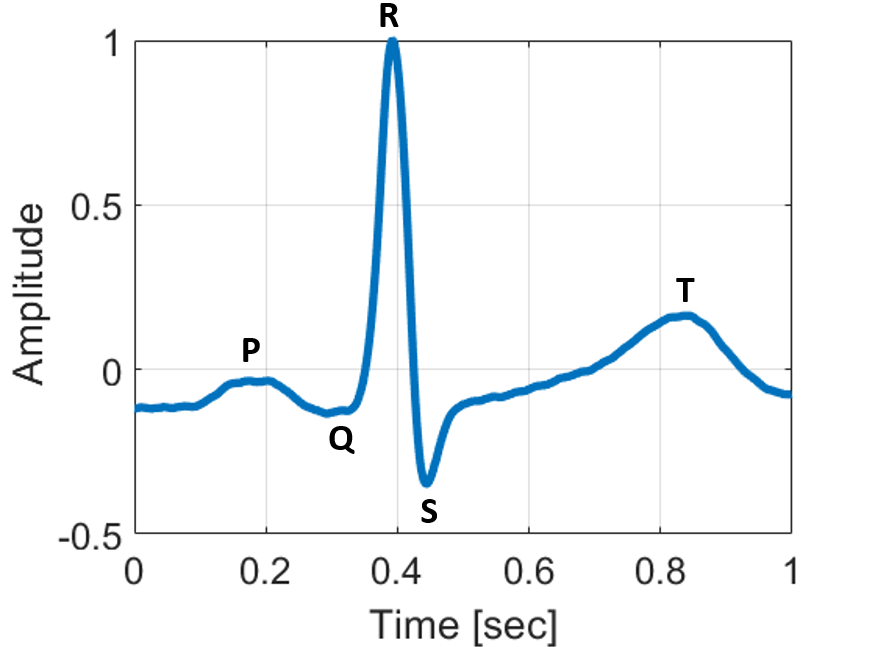}
		\caption{Example of standard ECG signal which is composed of five types of pulses, labeled P, Q, R, S, T.}
		\label{fig:ECG_COMP}
\end{figure}

Cardiovascular diseases represent a formidable global health challenge, necessitating prompt detection and effective management of heart-related conditions \cite{WHO,kotseva2021primary}. Electrocardiogram (ECG) and heart rate monitoring (HRM) techniques play an indispensable role in identifying and managing potentially fatal cardiovascular indicators. Consequently, the development of personalized and tailored HRM devices is of paramount importance. 

Current HRM methodologies, heavily reliant on continuous ECG monitoring, face significant hurdles due to the immense volume of data requiring processing, storage, and transmission \cite{majumder2018noncontact,wang2019energy}. Furthermore, portable ECG monitors must handle power consumption constraints, as higher sampling rates necessitate faster and more precise clocks during analog-to-digital conversion, inevitably leading to increased energy demands \cite{antony2018asynchronous,zhao2017random,yenduri2012low}. Addressing the limitations imposed by traditional synchronous clock-based samplers is therefore important. As a result, there is an urgent need for asynchronous, power-efficient samplers and sub-Nyquist sampling approaches capable of reducing power consumption and sampling rates without compromising the signal accuracy of the monitoring process, which are important for HRM.

Finite-rate-of-innovation (FRI) signals, commonly employed in time-of-flight applications like radar and ultrasound imaging, offer a powerful framework for sub-Nyquist sampling \cite{eldar2015sampling}. This approach, facilitated by tailor-made sampling kernels, ADCs, and parameter estimation blocks, enables accurate signal reconstruction at reduced sampling rates, thereby minimizing cost and power consumption. FRI theory showcases the efficacy of sampling parametric signals, such as ECGs, at rates equal to or higher than their innovation rate, ensuring precise capture of essential signal features \cite{vetterli2002sampling,blu2008sparse,eldar2015sampling,michaeli2011xampling}.

Sub-Nyquist FRI sampling techniques have been extensively studied for ECG signals. However, the use of fixed pulse width models faces challenges due to the varying shapes of ECG pulses. ECG signals consist of five distinct pulse types (see Fig. \ref{fig:ECG_COMP}), each corresponding to specific events in the cardiac cycle. Recent research has focused on alternative models, such as asymmetric pulses based on Gaussian, wavelet, or Lorentzian functions, aimed at reducing model mismatch error \cite{quick2012extension,sandryhaila2011efficient,rudresh2018asymmetric,nair2014p,hao2006compression}.

A variable-pulse-width (VPW) FRI model, introduced by Baechler et al. \cite{baechler2017sampling}, offers enhanced flexibility by incorporating pulse width and asymmetry parameters. However, the VPW-FRI model exhibits instability in noisy environments. Huang et al. \cite{huang2018optimization} attempted to improve the VPW-FRI using Particle Swarm Optimization (PSO), but encountered high time complexity in the reconstruction algorithm. In a subsequent study, Huang et al. \cite{huang2022sub} proposed an extension to the VPW-FRI model, employing signal differentiation to reduce model mismatch error for ECG signals. Nevertheless, this differentiation-based method compromises noise resistance. Additionally, the sampling rate of this method is $(R+1)$ times the minimal theoretical sampling rate, where $R$ is the highest order of the derivative, as it necessitates more parameters than the conventional VPW-FRI approach, potentially limiting its efficiency.

The aforementioned techniques all utilize clock-based samplers. Time-encoding ADCs, which are asynchronous samplers, present a unique and promising approach by encoding time intervals between events and directly converting them into digital signals. This method not only simplifies the ADC's design but also reduces power consumption and can minimize the required number of bits while ensuring accurate signal recovery \cite{naaman2021time}. 
The Integrate-and-Fire Time Encoding Machine (IF-TEM) is a time-based asynchronous sampler that eliminates the need for a global clock \cite{lazar2004time}. Its operation involves integrating an analog signal and comparing it with a threshold. Whenever the threshold is crossed, the time instances are recorded, effectively encoding the information contained in the analog signal \cite{lazar2004time}. The IF-TEM's ADC is specifically designed to be low-power, compact, and high-resolution, making it well-suited for continuous, ambulatory long-term monitoring applications \cite{nallathambi2015pulse,rastogi2011integrate}.
Sampling and perfect reconstruction of signals from TEM outputs have been extensively studied in the literature \cite{lazar2004time,adam2020sampling,alexandru2019reconstructing,naaman2021fri,rudresh2020time,gontier2014sampling,naaman2022uniqueness,naaman2021time,nallathambi2013integrate,nallathambi2015pulse}. The authors in \cite{rastogi2011integrate} proposed a low-power hardware implementation of IF-TEM suitable for ultra-low power body sensors. However, their approach maintains the sampling rate at least at the Nyquist rate, limiting the potential power savings.
To address this limitation, the authors in \cite{nallathambi2013integrate,nallathambi2015pulse} leveraged the sparse signal structure inherent in ECG signals, successfully reducing the IF-TEM sampling rate by a factor of four compared to the Nyquist rate. Nonetheless, their approach focuses only on detecting the middle part of each heartbeat, namely the QRS complex. 

\begin{figure}[t!]
	\vspace{-0.5cm}
		\centering
		\includegraphics[width = 0.5\textwidth]{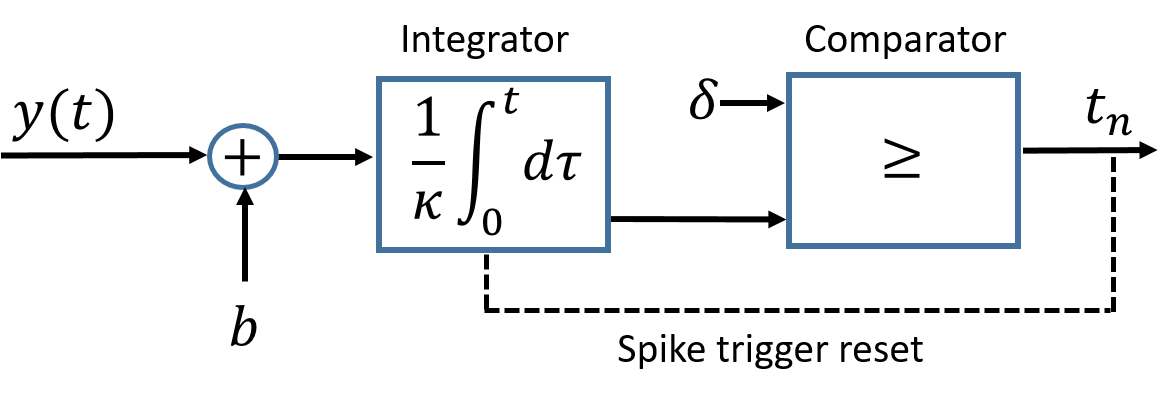}
		\vspace{-0.5cm}
		\caption{Time encoding machine with spike trigger reset. The input is biased by $b$, scaled by $\kappa$, and integrated. A time instant is recorded when the threshold $\delta$ is reached, after which the value of the integrator resets.}
		\label{fig:Vetterli_TEM}
	\end{figure}

In this paper, we present a robust time-based sub-Nyquist sampling and recovery algorithm for analog ECG signals using the IF-TEM ADC. Our approach uses a sampling kernel with enhanced noise resilience and applying a frequency-domain signal reconstruction method based on \cite{naaman2021fri,vetterli2002sampling,tur2011innovation}. We validate our method by applying it to heart rate monitoring (HRM) calculations based on the recovered ECG signal. The results demonstrate high accuracy compared to known HRM evaluations obtained using synchronous ADCs and compared to existing methods \cite{baechler2017sampling,huang2022sub}. Specifically, our considered sampling rate is only twice higher than the minimal theoretical rate proposed by \cite{baechler2017sampling}, and 66.66 times less than the Nyquist rate of 2000 Hz for synchronous ADCs \cite{schellenberger2020dataset}. Our study exhibits superior precision for HRM due to enhanced ECG reconstruction compared to the techniques proposed by \cite{baechler2017sampling,huang2022sub}, as evidenced by various statistical metrics.

To bridge the gap between theory and practice, we developed an IF-TEM ADC, providing a hardware validation of time-based sub-Nyquist sampling and reconstruction of noisy ECG signals at only twice the minimal theoretical rate. Our IF-TEM ADC differs from the one presented in \cite{naaman2023hardware} by incorporating a different filter necessary for ECG reconstruction and a distinct integrator implementation. Our algorithm and hardware successfully monitor and recover ECG signals while operating at significantly lower rates (20-40 times) than the Nyquist rate. The hardware experiments using our IF-TEM sampler demonstrate the empirical robustness of our method in a realistic, noisy setting, thereby validating its practical utility.

This paper is organized as follows. In Section \ref{sec:TEM_and_problem} we review IF-TEM, followed by the problem formulation. Section \ref{sec:theory} presents our main result for robust sampling and recovery of ECG signals from the IF-TEM measurements and simulation results. 
Our hardware prototype is demonstrated in Section \ref{sec:HW_spec}.
Concluding remarks are presented in Section \ref{sec:Conclusions}.

\section{Preliminaries and Problem Formulation}
\label{sec:TEM_and_problem}
\subsection{Time Encoding Machine}
\label{subsec:TEM}
IF-TEM is an asynchronous ADC that operates by encoding the input analog signal into a sequence of time instants or firing times. As illustrated in Fig. \ref{fig:Vetterli_TEM} and described in \cite{lazar2004perfect}, the IF-TEM takes a bounded analog signal $y(t)$ as input, where $|y(t)| \leq c$ for some positive constant $c$.
The IF-TEM is characterized by three positive real-valued parameters: $b$, $\kappa$, and $\delta$. The input signal $y(t)$ is first offset by a bias value $b$, where $b > c$, ensuring that the biased signal $y(t) + b$ is strictly positive. This biased signal is then integrated and scaled by a factor $\kappa$. Whenever the integrated and scaled signal reaches a predetermined threshold $\delta$, the corresponding time instant $t_n$ is recorded, and the integrator is reset to zero.
This process repeats, with the IF-TEM capturing the next time instant $t_{n+1}$ such that the integral of the biased and scaled input signal over the interval $[t_n, t_{n+1}]$ equals the threshold $\delta$, as expressed by
\begin{align}
   \frac{1}{\kappa}\int_{t_n}^{t_{n+1}} (y(s)+b)\, ds = \delta.
   \label{eq:y_t_relation_st}
\end{align}
The time encodings ${t_n, n \in \mathbb{Z}}$ obtained from the IF-TEM constitute a discrete representation of the continuous-time analog input signal $y(t)$. The objective is to reconstruct the original analog signal $y(t)$ from these time encodings. Typically, the reconstruction process involves utilizing an alternative discrete representation $\{y_n, n \in \mathbb{Z}\}$, defined as \cite{lazar2004perfect,naaman2021fri,naaman2022uniqueness,adam2019multi,alexandru2020sampling}:
\begin{equation}
y_n \triangleq \int_{t_n}^{t_{n+1}}y(s)\, ds =  -b(t_{n+1}-t_n)+\kappa\delta.
\label{eq:trigger0}
\end{equation}
These measurements $\{y_n, n \in \mathbb{Z}\}$ are derived from the time encodings $\{t_n, n \in \mathbb{Z}\}$ and the IF-TEM parameters $\{b, \kappa, \delta\}$.

While the reconstruction methods vary for different signal classes, perfect recovery of any signal requires that the firing rate, determined by the time encodings, satisfies a lower bound that depends on the degrees of freedom of the signal \cite{lazar2004perfect}. Notably, the firing rate of an IF-TEM is bounded both from above and below, with the bounds being functions of the IF-TEM parameters and an upper bound on the signal amplitude.
By leveraging \eqref{eq:trigger0} and the fact that $|y(t)| \leq c$, it can be shown that for any two consecutive time instances \cite{lazar2004perfect}:
\begin{equation}
\frac{\kappa\delta}{b+c} \leq t_{n+1} - t_n \leq \frac{\kappa\delta}{b-c}.
\label{eq:consecutive_time}
\end{equation}

\label{subsec:problem}

\begin{figure*}[hbt!]
	\vspace{-0.5cm}
		\centering
		\includegraphics[width = 1\textwidth]{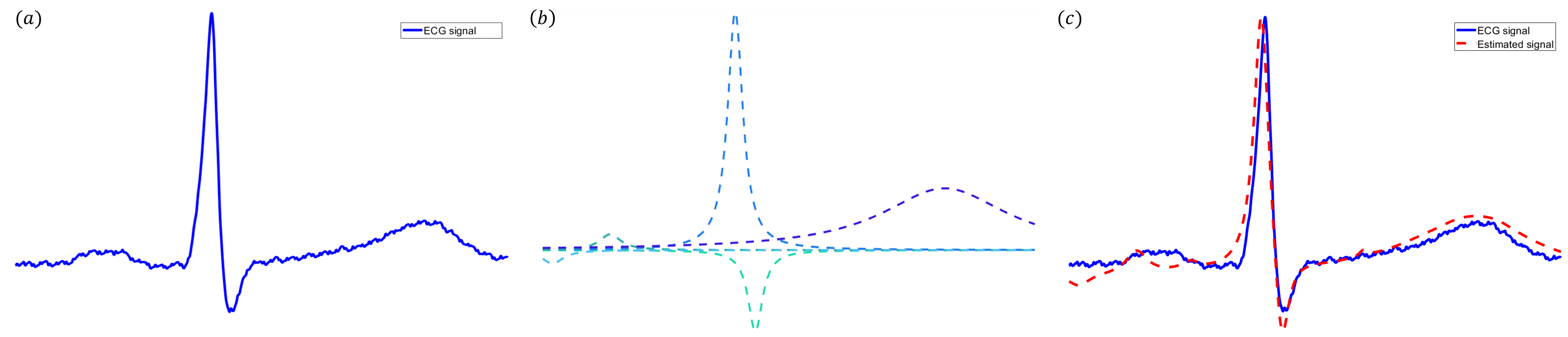}
		\vspace{-0.5cm}
		\caption{Illustration showcasing ECG signal decomposition into VPW-FRI asymmetric pulses: (a) a single pulse cycle of one second, (b) its division into five distinct asymmetric VPW-FRI pulses; (c) the reconstructed ECG signal is achieved through the summation  of the five pulses.}
		\label{fig:decompose}
\end{figure*}

\begin{figure}[h!]
    \centering
    \includegraphics[width= 0.5\textwidth]{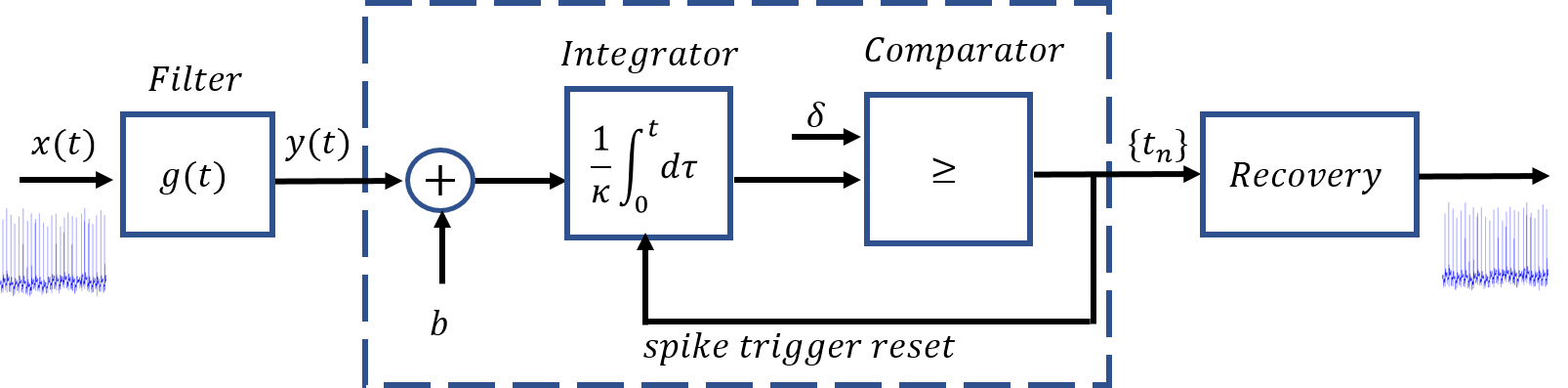}
    \caption{IF-TEM Sampling: Continuous-time ECG signal $x(t)$ is filtered through a sampling kernel $g(t)$ and then sampled using an IF-TEM to generate time-instants $\{t_n\}$ from which the ECG signal is recovered.}
    \label{fig:TEMfri}
\end{figure}

\subsection{Problem Formulation}
As illustrated in Fig. \ref{fig:ECG_COMP}, an ECG signal is composed of a series of pulses, which aligns with the concept of a pulse train signal. Pulse train signals fall within the category of FRI signals.
This presents an opportunity to employ the principles of FRI sampling theory for sub-Nyquist sampling of ECG signals. Compliant with FRI theory, an ECG signal lends itself to representation through a weighted combination of pulse functions.
These functions directly correspond to the characteristic forms of the P, Q, R, S, and T waves embedded within the ECG waveform; each corresponds to a specific event within the cardiac cycle \cite{goldberger1981clinical}.
In Fig. \ref{fig:decompose}, we provide an illustration depicting the decomposition of an ECG signal using five VPW-FRI asymmetric pulses.

We consider an ECG signal which is presented by a $T$-periodic VPW-FRI signal of the form \cite{baechler2017sampling}:
\begin{equation}
    x(t) = \sum_{k=0}^{K-1}x_k(t), \quad\text{where,} \quad x_k(t) = x_k^s(t)+x_k^a(t),
    \label{eq:fri}
\end{equation}
\begin{equation}
    x_k^s(t) = c_k\sum_{p\in\mathbb{Z}}\frac{r_k}{\pi\left( r_k^2 +\left(t-T_k-pT\right)^2\right)},
\end{equation}
\begin{equation}
    x_k^a(t) = d_k\sum_{p\in\mathbb{Z}}\frac{t-T_k-pT}{\pi\left( r_k^2 +\left(t-T_k-pT\right)^2\right)}.
\end{equation}
In this model, each VPW-FRI pulse $x_k(t)$ is characterized by four parameters: $c_k$ (symmetric amplitude), $d_k$ (asymmetric amplitude), $r_k$ (pulse width), and $T_k$ (temporal delay). The signal components $x_k^s(t)$ and $x_k^a(t)$ represent the symmetric and asymmetric parts of the pulse, respectively.
Theoretically, an ECG signal can be represented using the VPW-FRI model with $K=5$ pulses, corresponding to the five characteristic waveforms.
VPW-FRI pulses can be seen as an extension of the FRI model. By allowing the asymmetrical amplitude parameter $d_k$ to be zero and taking the limit of $x_k(t)$ as the pulse width $r_k$ approaches zero, the result is equivalent to a Dirac delta with an amplitude of $c_k$ at the time delay $T_k$.

To sample $x(t)$ at a sub-Nyquist rate, we first pass it through a designed sampling kernel $g(t)$ and then measure low-rate samples of the filtered signal $y(t)$ using an IF-TEM ADC \cite{naaman2023hardware}. The sampling kernel $g(t)$ should be designed such that the VPW-FRI parameters $\{T_k, r_k, c_k, d_k\}_{k=0}^{K-1}$ can be accurately computed from the IF-TEM samples, as depicted in Fig. \ref{fig:Vetterli_TEM}.
In particular, it has been shown that using a classical clock-based uniform sampler, $4K$  samples of $x(t)$ in an interval $T$ are sufficient to determine $\{T_k,r_k,c_k,d_k\}_{k=0}^{K-1}$ uniquely \cite{baechler2017sampling}.
The problem at hand is the perfect recovery of the ECG signal's VPW-FRI parameters $\{T_k, r_k, c_k, d_k\}_{k=0}^{K-1}$ using an IF-TEM sampling scheme, as shown in Fig.~\ref{fig:TEMfri}. Specifically, we aim to design the sampling kernel $g(t)$ and the IF-TEM parameters ${b, \kappa, \delta}$ such that the VPW-FRI parameters are uniquely determined from the time-encodings obtained from the IF-TEM. Additionally, we develop a reconstruction algorithm capable of accurately recovering the ECG signal based on the time-encoded measurements.

\section{ECG-TEM: sampling and perfect recovery of VPW-FRI Signals from IF-TEM measurements}
\label{sec:theory}
In this section, we introduce a method to perfectly recover VPW-FRI signals, which model ECG signals, from measurements obtained using the IF-TEM. We leverage the fact that the ECG signal $x(t)$ in \eqref{eq:fri} can be perfectly reconstructed from its $4K$ Fourier series coefficients (FSCs) \cite{baechler2017sampling}. We derive conditions on the IF-TEM parameters and the sampling kernel $g(t)$ such that the $4K$ FSCs of the input VPW-FRI signal are uniquely recovered from the IF-TEM output.
\vspace{-0.62mm}
\subsection{Fourier-Series Representation of VPW-FRI Signals}
\label{sub:Fourier}
We begin by explicitly relating the ECG input signal $x(t)$ of \eqref{eq:fri} to its FSCs following \cite{baechler2017sampling}.

Given that $x(t) = \sum_{k=0}^{K-1}x_k(t)$ in \eqref{eq:fri} forms a signal with a periodicity of $T$ (each $x_k(t)$ has period $T$), it has a Fourier series representation
\begin{equation}
    x(t) = \sum_{m\in\mathbb{Z}}X[m]e^{jk\omega_0 t},
    \label{eq:x_initial}
\end{equation}
where $\omega_0=\frac{2\pi}{T}$. The Fourier-series coefficients $X[m]$ are given by 
\begin{equation} \label{eq:x_fourier}
\begin{split}
X[m]&= \sum_{k=1}^K X_k^s[m]+X_k^a[m]\\
&=\sum_{k=1}^K\frac{c_k-jd_k sgn(m)}{T}e^{-2\pi\left(r_k|m|+jT_km\right)/T},
\end{split}
\end{equation}
where, 
\begin{equation}
    X_k^s[m] = \frac{c_k}{T}e^{-2\pi\left(r_k|m|+jT_km\right)/T},
\end{equation}
represents the symmetric component, and
\begin{equation}
    X_k^a[m] = -\frac{jd_k}{T}sgn(m)e^{-2\pi\left(r_k|m|+jT_km\right)/T},
\end{equation}
represents the anti-symmetric component of the FSCs.
Since $x(t)$ is real-valued, its FSCs $X[m]$ are complex conjugate pairs, that is,
\begin{align}
X^*[-m] = X[m].
\label{eq:complex_conj}
\end{align}
According to \cite{baechler2017sampling,huang2022sub}, only the positive indices $m \geq 0$ are considered for the VPW-FRI spectrum. This restriction arises due to the presence of a cusp at $m=0$, which prevents the annihilation of both positive and negative spectrum values owing to the decaying nature of the VPW spectrum.

The sequence outlined in \eqref{eq:x_fourier} represents a spectral estimation problem. The parameters $\{u_k, v_k\}_{k=1}^K$ can be uniquely estimated employing high-resolution spectral estimation theory \cite{stoica}. A well-established technique like the annihilating filter (AF) method \cite{vetterli2002sampling} can be utilized to compute $\{u_k, v_k\}_{k=1}^K$. A series of $2K$ consecutive values of $X[m]$ needs to be computed to determine these parameters. 

In the context of the VPW-FRI framework, if the number $K$ of VPW-FRI pulses is known, the $4K$ unknown parameters in \eqref{eq:fri}, denoted as $\{c_k, d_k, r_k, T_k\}_{k=1}^K$, can be estimated through a parameter estimation algorithm that aligns with the FSCs of the ECG signal.
Consequently, our task is reduced to the distinct determination of the desired number of FSCs from the signal measurements. Given that $x(t)$ typically comprises a substantial number of FSCs, we discuss next a sampling kernel design that removes unnecessary FSCs and thus reduces the sampling rate.

\subsection{Sampling Kernel}
\label{sub:Kernel}

Since a minimum of $4K$ FSCs are sufficient for uniquely recovering the ECG signal, the sampling kernel $g(t)$ is designed to remove or annihilate any additional FSCs. The filtered signal $y(t)$ is given by \cite{naaman2021fri}
\begin{equation}
\begin{split}
    y(t) &= (x * g)(t)= \int_{-\infty}^{\infty} x(\tau)g(t-\tau)d\tau\\
    &=\sum_{m\in \mathbb{Z}}
    X[m]\int_{-\infty}^{\infty} g(t-\tau)e^{jm\omega_0 \tau}d\tau\\
    &= \sum_{m\in \mathbb{Z}}
   X[m]\, \hat{g}(m\omega_o) \,e^{jm\omega_0 t}.\\
    \end{split}
     \label{eq:yt_by_x}
\end{equation}

Following the approach proposed in \cite{naaman2021fri,naaman2023hardware}, we define the sampling kernel $g(t)$ to exclude the zeroth Fourier coefficient (DC component) of the filtered signal $y(t)$, leading to a robust reconstruction process. The kernel is designed to satisfy the following condition in the Fourier domain:

\begin{equation}
\hat{g}(m\omega_0) =
\begin{cases}
1 & \text{if $m\in \mathcal{M}$},\\
0 & \text{otherwise},
\end{cases}
\label{eq:kernel}
\end{equation}
where
\begin{equation}
    \mathcal{M} = \{-M,\cdots,-1,1,\cdots,M\}.
    \label{eq:M}
\end{equation}
Here, $\mathcal{M}$ is a set of integers such that $|\mathcal{M}|\geq 4K$, ensuring that at least $4K$ FSCs are preserved for unique recovery of the VPW-FRI parameters. Sampling kernels that fulfill the criteria outlined in \eqref{eq:kernel} include the sinc function as discussed in \cite{vetterli2002sampling}, exponential and polynomial reproducing kernels as explored in \cite{dragotti2007sampling},  sum-of-modulated spline kernels as described in \cite{mulleti2017paley}, 
sum-of-sincs (SoS) kernel as presented in \cite{tur2011innovation} and more.

Note that an ideal lowpass filter with an appropriate cutoff frequency can also be applied to remove the FSCs. Following the kernel design in \eqref{eq:kernel}, the filtered signal $y(t)$ takes the form:
\begin{equation}
y(t) =\sum_{m\in \mathcal{M}} X[m]\hat{g}(m\omega_0) e^{jm\omega_0 t}=\sum_{m\in \mathcal{M}} X[m] e^{jm\omega_0 t}.
\label{eq:yt_by_x22_M}
\end{equation}
The filtered signal $y(t)$ is sampled by an IF-TEM, which requires its input to be real-valued and bounded. The boundedness of $y(t)$ is established in Appendix \ref{app:a}.

\subsection{ECG-TEM Sampling and Recovery Guarantees}
\label{sub:Samp}
The IF-TEM input is the filtered signal $y(t)$, which is the $T$-periodic signal defined in \eqref{eq:yt_by_x22_M}.
The output of the IF-TEM is a set of time instants $\{ t_n\}_{n\in\mathbb{Z}}$. Given $\{t_n\}$ one can determine the measurements $\{y_n\}$ by using \eqref{eq:trigger0}.
The relation between the measurements $y_n$ and the desired FSCs is given by 
\begin{equation}
    \begin{split}
        y_n &= \int_{t_n}^{t_{n+1}}y(t)\, dt\\ &= \int_{t_n}^{t_{n+1}} \sum_{m\in \mathcal{M}}
   X[m]e^{jk\omega_0 t}dt\\
    &= \sum_{m\in \mathcal{M} } X[m] \frac{\left( e^{jm\omega_0 t_{n+1}}- e^{jm\omega_0 t_{n}} \right)}{jm\omega_0}.
    \end{split}
    \label{eq:yx_rel1}
\end{equation}

Partial summation of IF-TEM measurements is utilized to enhance the robustness of signal reconstruction, as shown by \cite{naaman2023hardware}, which proves more effective than directly using the raw IF-TEM measurements.
The partial sums of the measurements $y_n$ are defined as
\begin{equation}
    z_n=\sum_{i=1}^{n-1} y_i = \sum_{m\in \mathcal{M} } \frac{X[m]}{jm\omega_0} \left( e^{jm\omega_0 t_{n}}- e^{jm\omega_0 t_{1}} \right),
    \label{eq:partial}
\end{equation}
where $ n = 2,\cdots,N$.
Note that \eqref{eq:partial} can be alternatively expressed as
\begin{equation}
    z_n= \sum_{m\in \mathcal{M} } \frac{X[m]}{jm\omega_0}  e^{jm\omega_0 t_{n}}+ c,
    \label{eq:partial2}
\end{equation}
where 
\begin{equation}
     c = - \sum_{m\in \mathcal{M} } \frac{X[m]}{jm\omega_0}  e^{jm\omega_0 t_{1}}.
\end{equation}
Let $\mathbf{z} = [z_2,\cdots,z_N]^\mathrm{T}\in\mathbb{R}^{N-1}$ denote the the vector of partial sums, and let $\mathbf{\hat{z}}$ be the vector of FSCs, with $c$ in the zeroth place:
\begin{equation}
    \mathbf{\hat{z}} =
\left[-\frac{X[-M]}{jM\omega_0}, \cdots, -\frac{X[-1]}{j\omega_0},\kern 0.16667em c \kern 0.16667em,\frac{X[1]}{j\omega_0},\cdots, \frac{X[M]}{jM\omega_0}\right]^{\top}.
\label{eq:z_hat}
\end{equation}
Define  $\mathbf{B}\in\mathbb{C}^{(N-1) \times (2M+1)}$ as the matrix
\begin{equation}
    \mathbf{B} = {\begin{bmatrix}e^{-jM\omega_0 t_2}& \cdots1\cdots&{e^{jM\omega_0 t_2}}\\
     e^{-jM\omega_0 t_3}&\cdots1\cdots&{e^{jM\omega_0 t_3}}\\ 
     \vdots &\ddots&\vdots \\ 
     e^{-jM\omega_0 t_N}& \cdots1\cdots&{e^{jM\omega_0 t_N}}
    \end{bmatrix}}.
    \label{eq:matrixB}
\end{equation}
Then, \eqref{eq:partial2} can be expressed in matrix form as follows:
\begin{equation}
     \mathbf{z} = \mathbf{B}\, \mathbf{\hat{z}}.
     \label{eq:forward_model}
\end{equation}


In \cite{naaman2023hardware}, it is established that, given the set of distinct time instants $\{t_n\}_{n=2}^N$ and the Vandermonde structure of the matrix $\mathbf{B}$, the condition $N-1 \geq 2M+1$ ensures that $\mathbf{B}$ has full column rank. Consequently, $\mathbf{B}$ is left-invertible, enabling perfect reconstruction of the FSCs vector $\mathbf{\hat{z}}$ via
\begin{equation}
    \mathbf{\hat{z}} = \mathbf{B}^\dagger\,\mathbf{z},
    \label{eq:fourierZ}
\end{equation}
where $\mathbf{B}^\dagger=\left(\mathbf{B}^\top\mathbf{B}\right)^{-1}\mathbf{B}^\top$ denotes the Moore-Penrose inverse of $\mathbf{B}$. Once $\mathbf{\hat{z}}$ is obtained, the FSCs $\hat{x}[k]$ are uniquely determined.  Using the relation
\begin{align}
    \hat{z}[m]= \begin{cases}
     \frac{X[m]}{j\omega_0m}, & \text{if $m\in\mathcal{M}$ },\\
      -\sum_{m^\prime\in \mathcal{M}}\left( \frac{X[m^\prime]}{jm^\prime\omega_0} \right)  e^{jkm^\prime\omega_0 t_1} & \text{if $m=0$ },
      \end{cases}
      \vspace{.1in}
      \label{eq:get_x}
\end{align}
we have the vector of FSCs $\mathbf{\hat{x}}$: 
\begin{equation}
    \mathbf{\hat{x}} = \left[\hat{z}[-M],\cdots,\hat{z}[-1],\hat{z}[1],\cdots,\hat{z}[M]\right]^{\top} \in\mathbb{C}^{2M}.
    \label{eq:zandx0}
\end{equation}
By employing the vector $\mathbf{\hat{z}}$ and the relation in \eqref{eq:zandx0}, the vector of FSCs $\mathbf{\hat{x}}$ is uniquely determined. This indicates that, within the modified kernel configuration and in the absence of the zero frequency, the set of FSCs ${\hat{x}[k]}$ can be uniquely determined from the time encodings if $N-1 \geq 2M+1$. This requirement implies a minimum of $2M+2$ firing instants within the interval $T$.
We next show that for the sampling kernel choice \eqref{eq:kernel}, we can uniquely identify an ECG signal described as a VPW-FRI signal, from the IF-TEM time instances. Our results are summarized in the following theorem.
\begin{theorem}
\label{theorem:FRI0}
Let $x(t)$ be an ECG signal, described by a $T$-periodic VPW-FRI model of the  form $ x(t) = \sum_{k=0}^{K-1}x_k(t)$,
as defined in \eqref{eq:fri}. Consider the sampling mechanism shown in Fig.~\ref{fig:TEMfri}. Let the sampling kernel $g(t)$ satisfy
\begin{align*}
    \hat{g}(m\omega_0) =       
    \begin{cases}
      1 & \text{if $m\in \mathcal{M}=\{-M,\cdots,-1,1,\cdots,M \}$},\\
      0 & \text{otherwise}.
    \end{cases}
    \end{align*}
Choose the real positive TEM parameters $\{b, \kappa, \delta\}$ such that $c<b<\infty$, and
\begin{equation}   
\frac{b-c}{\kappa\delta} \geq \frac{8K+2}{T}. \label{eq:sample_bound0}
\end{equation} 
Then, the ECG parameters $\{T_k,r_k,c_k,d_k\}_{k=0}^{K-1}$ can be perfectly recovered from the IF-TEM outputs if $M \geq 4K$.
\end{theorem}
\begin{proof}
Consider a positive integer $K$ and a number $T>0$. Let $0\leq t_1<t_2<\cdots<t_N < T$ for an integer $N$, and $\omega_0 = \frac{2\pi}{T}$. 
The IF-TEM input is the filtered signal $y(t)$, which is the $T$-periodic signal defined in \eqref{eq:yt_by_x22_M}.
The output of the IF-TEM is a set of time instances $\{ t_n\}_{n\in\mathbb{Z}}$. Given $\{t_n\}$ one can determine the measurements $\{y_n\}$ by using \eqref{eq:trigger0}.
The relation between the measurements $y_n$, their partial sum, and the desired FSCs are defined in \eqref{eq:yx_rel1} and \eqref{eq:partial}, respectively. 
The Fourier coefficients $\{X[m]\}_{m\in \mathcal{M}}$ are uniquely determined using \eqref{eq:fourierZ} and \eqref{eq:zandx0}, provided that there are at least $N \geq 2M+2$ time instances $\{t_n\}_{n=1}^N$ in an interval $T$, where $M$ should be at least the number of degrees of freedom of the signal $x(t)$. 
This implies that $M\geq 4K$ and there should be a minimum of $8K+2$ IF-TEM time instances within an interval of $T$ to enable recovery of the FSCs and subsequent reconstruction of the VPW-FRI signal. To ensure this, the IF-TEM parameters are chosen such that 
\begin{align}
    \frac{b-c}{\kappa\delta} \geq \frac{8K+2}{T}.
    \label{eq:minFR}
\end{align}

After computing the FSCs ${X[m]}_{m\in \mathcal{M}}$ from \eqref{eq:zandx0}, using \eqref{eq:x_fourier}, we have:
\begin{equation} \label{eq:x_fourier2}
\begin{split}
X[m] &=\sum_{k=1}^K \frac{c_k-jd_k}{T}e^{-2\pi\left((r_k+jT_k)m\right)/T}\\
&= \sum_{k=1}^K v_k u_k^m\quad m\in\mathcal{M},
\end{split}
\end{equation}
where $v_k=\frac{c_k-jd_k}{T}$ and $u_k=e^{\left(-2\pi (r_k+j T_k)\right)/T}$.
The reconstruction of the parameters $\{c_k, d_k, r_k, T_k\}_{k=1}^K$ becomes a spectrum estimation problem. In the absence of noise, Prony's method can be employed for perfect parameter estimation when $M\geq 4K$. Prony's method involves constructing the unique annihilating filter for the FSCs.

The unique annihilating filter $A(z)$ for the revised FSCs $X[m]$ in \eqref{eq:x_fourier} is given by:
\begin{equation}\label{eq:ann_uniq}
A(z) = \sum_{k=0}^{K} A[k]z^{-k} = \prod_{k=0}^{K-1} (1-u_kz^{-1}).
\end{equation}
The convolution of $A(z)$ with $X[m]$ satisfies:
\begin{equation}\label{eq:annhilating}
\begin{split}
(A*X)[m] &= \sum_{l=0}^{K} X[m-l]\\
&= \sum_{k=0}^{K-1} (c_k-jd_k)\left(\sum_{l=0}^{K} A[l]u_k^{-l} \right)u_k^m = 0.
\end{split}
\end{equation}
Given the roots $\{u_k\}_{k=0}^{K-1}$ of the annihilating filter $A(z)$, the delays $T_k$ can be computed using $T_k=-\frac{T\angle u_k}{2\pi}$, and the widths $r_k$ can be computed as $r_k=\frac{T\log |u_k|}{2\pi}$.
Finally, the parameters $c_k$ and $d_k$ are then retrieved by solving \eqref{eq:x_fourier2}, with $c_k = T \operatorname{Re}(v_k)$ and $d_k = T \operatorname{Im}(v_k)$, completing the proof.
\end{proof}
Based on Theorem~\ref{theorem:FRI0}, a reconstruction algorithm to compute the VPW-FRI parameters from IF-TEM firings is presented in Algorithm~\ref{alg:algorithm_th1}.
\begin{algorithm}[!h]
\caption{Reconstruction of an ECG signal using an IF-TEM sampler}
\label{alg:algorithm_th1}
\begin{algorithmic}[1]
\Statex \textbf{Input:} $N\geq 8K+2$ spike times $\{t_n\}_{n=1}^N$ in a period $T$.
\State $n \gets 1$
\While{$n\leq N-1$}
\State Compute $y_n = -b(t{n+1}-t_n)+\kappa\delta$
\State $n\gets n+1$
\EndWhile
\State Compute $\mathbf{z} = [z_2,\cdots,z_N]^{\mathrm{T}}\in\mathbb{R}^{N-1}$ using \eqref{eq:partial}
\State Compute $\mathbf{B}$ using \eqref{eq:matrixB}
\State Compute the Fourier coefficients vector $\mathbf{\hat{z}} = \mathbf{B}^\dagger\mathbf{z}$
\State Compute the Fourier coefficients vector $\mathbf{\hat{x}} = [\hat{z}[-M],\cdots,\hat{z}[-1],\hat{z}[1],\cdots,\hat{z}[M]]^{\top} \in\mathbb{C}^{2M}$ using \eqref{eq:zandx0}
\State Construct the unique annihilating filter $A$ for the FSCs $X[m],\quad m\in\mathcal{M}$ using \eqref{eq:ann_uniq}
\State Denoise the Fourier coefficients vector $\mathbf{\hat{x}}$ (see Section \ref{sec:denoise})
\State Compute the roots $u_k=e^{\left(-2\pi (r_k+j T_k)\right)/T}$ of the annihilating filter $A(z)$ using \eqref{eq:annhilating}
\State Compute $T_k=-\frac{T\angle u_k}{2\pi}$
\State Compute $r_k=\frac{T\log |u_k|}{2\pi}$
\State Compute $v_k=\frac{c_k-jd_k}{T}$ by solving linearly \eqref{eq:x_fourier2} with $u_k$
\State Compute the amplitudes $c_k = T \operatorname{Re}(v_k)$
\State Compute the amplitudes $d_k = T \operatorname{Im}(v_k)$
\State \textbf{return} $\{T_k,r_k,c_k,d_k\}_{k=0}^{K-1}$
\Statex \textbf{Output:} $\{T_k,r_k,c_k,d_k\}_{k=0}^{K-1}$
\end{algorithmic}
\end{algorithm}

\subsection{IF-TEM Parameter Selection}
\label{subsec:params}
The IF-TEM parameters are selected such that there is a minimum of $N\geq 2M+2$ time instants $\{t_n\}_{n=1}^N$ within a time interval $T$, where $M\geq 4K$.
Thus, the minimum firing rate that enables accurate reconstruction is $\frac{8K+2}{T}$. The maximum firing rate is bounded by $\frac{b+c}{\kappa\delta}$.
While the threshold $\delta$, which is a parameter of the comparator, is easier to control, the integrator constant $\kappa$ is a parameter of the integrator, and it is usually fixed. Thus, assuming a fixed value of $b$ and $\kappa$, choosing small $\delta$ results in a large firing rate above the minimum desirable value of $\frac{8K+2}{T}$.
In practice, both $b$ and $\delta$ are generated through a DC voltage source, and therefore large values of bias and threshold require high power. Hence, to minimize the power requirements, it is desirable for $b$ and $\delta$ to be as small as possible.

\subsection{Denoisers}
\label{sec:denoise}
In practical scenarios, the acquired signals are often corrupted by noise during the measurement process. Even if the reconstruction model is accurate, the acquisition devices themselves introduce noise, degrading signal quality. This highlights the importance of incorporating a denoising step to mitigate the effects of noise on the FSCs before estimating the ECG signal parameters. We assume an additive white Gaussian noise (AWGN) model with zero mean and independent and identically distributed (i.i.d.) samples. Moreover, we consider the noise to be introduced after the IF-TEM sampler acquisition stage.

\begin{figure}[t]
\centering
\includegraphics[width=0.5\textwidth]{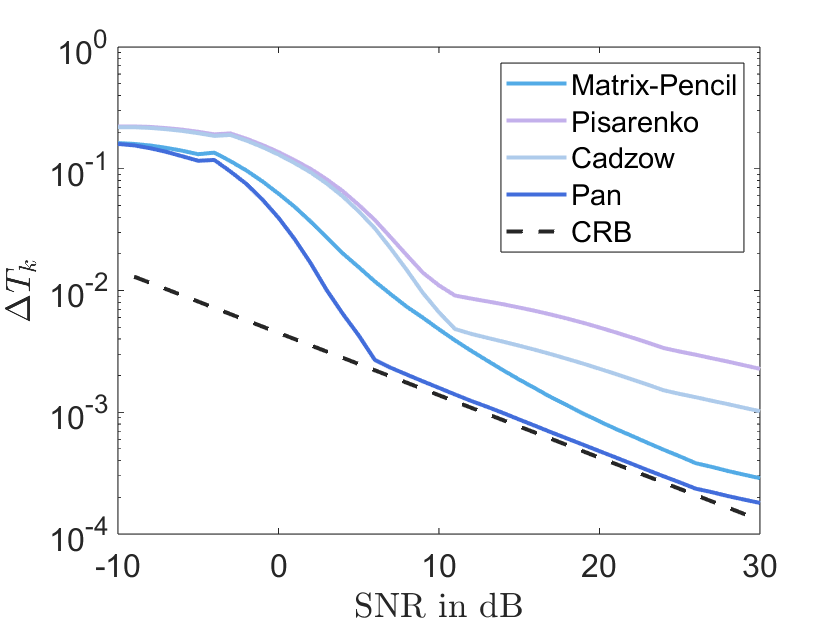}
\caption{Comparison of denoising approaches for estimating the location of VPW-FRI pulses. The estimation procedure utilizes 10 Fourier coefficients, with results averaged over 5,000 repetitions. For the single VPW-FRI pulse case, the Cramér-Rao bound on location estimation variance is plotted (dashed line).}
\label{fig:denoisers}
\end{figure}

Assume that AWGN is present. The Fourier series coefficients $X[m]$ can then be expressed as:
\begin{equation}
\tilde{X}[m] = X[m] + \zeta[m],
\end{equation}
where $\zeta[m]$ represents the FSCs of the AWGN.
Thus, the FSCs $X[m]$ contain a noise component, resulting in the annihilating filter condition in \eqref{eq:annhilating} no longer being strictly satisfied:
\begin{equation}\label{eq:annhilating2}
(A * \tilde{\mathbf{x}})[m] \approx 0,
\end{equation}
where $A(z)$ is the annihilating filter for the noise-free FSCs $\mathbf{\hat{x}}$.

To minimize the approximation error in \eqref{eq:annhilating2}, a higher sampling rate is necessary to obtain more Fourier coefficients and construct a larger Toeplitz matrix for estimating the parameters $u_k$ \cite{blu2008sparse}. Additionally, applying denoising techniques to the observed FSCs $\tilde{\mathbf{x}}$ can improve the accuracy of the reconstruction by reducing noise interference.

Several denoising techniques have been considered for the task of recovering the parameters of the VPW-FRI pulses from noisy measurements, as illustrated in Fig. \ref{fig:denoisers}. One widely used method is the Cadzow denoiser \cite{cadzow1988signal}. The annihilating filter condition in \eqref{eq:annhilating} can be expressed in matrix form as:
\begin{equation}
\mathbf{S}\mathbf{h} = \mathbf{0},
\end{equation}
where $\mathbf{S}$ is a rank-deficient Toeplitz matrix formed from consecutive values of the noisy DFT coefficients $\tilde{\mathbf{x}}$, and $\mathbf{h}$ is a vector containing the $(K+1)$ annihilating filter coefficients. The Cadzow algorithm is an iterative method that alternates between enforcing a rank of $K$ (the number of VPW-FRI pulses) and a Toeplitz structure on the noisy matrix $\tilde{\mathbf{S}}$. The low-rank approximation is achieved through the singular value decomposition (SVD) of $\tilde{\mathbf{S}}$, retaining only its $K$ largest singular values. To maintain the Toeplitz structure, diagonal averaging is employed.

Another technique is the matrix pencil method \cite{hua1990matrix}, also known as ESPRIT \cite{roy1989esprit}, which exploits the rotational invariance of the signal subspace and is non-iterative. Due to its non-iterative nature, the matrix pencil method can be applied sequentially after Cadzow denoising. Additionally, Pisarenko's method \cite{pisarenko1973retrieval} is considered. This technique estimates the annihilating filter by extracting the last column of the matrix $\mathbf{V}$ obtained from the SVD of the noisy Toeplitz matrix $\tilde{\mathbf{S}} = \mathbf{U}\boldsymbol{\Lambda}\mathbf{V}^*$. The last columns of $\mathbf{V}$ form an orthogonal basis for the nullspace of $\tilde{\mathbf{S}}$, and in this case, the nullspace of the original Toeplitz matrix $\mathbf{S}$ is one-dimensional. Pisarenko's method or the matrix pencil method can be employed sequentially after applying the Cadzow denoiser.

The final method examined is the Pan denoiser, inspired by IQML \cite{pan2016towards}. The annihilating filter $\mathbf{h}$ is derived iteratively by solving the following minimization problem:
\begin{equation}
\underset{\hat{\mathbf{x}},\mathbf{h}}{\text{min}} \left| \tilde{\mathbf{f}} - \mathbf{L}\hat{\mathbf{x}} \right|_2^2,
\end{equation}
subject to $\mathbf{h}^* \hat{\mathbf{x}} = \mathbf{0}$, where $\mathbf{L}$ is a linear transformation that maps the annihilatable signal to the measurements, and $\mathbf{h}$ contains the $(K+1)$ annihilating filter coefficients.  In the case of VPW-FRI with uniform sampling, $\mathbf{L}$ transforms the denoised DFT (Discrete Fourier Transform) coefficients $\hat{\mathbf{x}}$ (corresponding to positive frequencies) to the discrete-time noisy measured signal $\tilde{\mathbf{f}}$.
When employing an IF-TEM sampler, the denoising process is performed in the frequency domain on the noisy FSCs $\tilde{X}[m]$. Unlike Pan denoiser approach which directly denoise the signal measurements, the IF-TEM approach only provides access to time measurements rather than the signal samples themselves. Therefore, we adapt a similar concept to the Pan denoiser technique, but instead of denoising the noisy signal measurements directly, we focus on denoising the noisy Fourier coefficients $\tilde{X}[m]$ iteratively. 

In Fig. \ref{fig:denoisers}, a comparison of different denoising methods' performance is presented for the VPW-FRI pulse under varying noise levels. The matrix pencil and Cadzow approaches exhibit similar performance characteristics. Notably, Pan's technique outperforms all the other methods by a substantial margin. Additionally, the Cramer-Rao bound (CRB), which provides a theoretical lower limit on the best achievable performance by any unbiased estimator, derived in \cite{baechler2017sampling}, is displayed. Based on our evaluations, we opted for the Pan denoiser due to its compelling performance improvements and direct application to measured signals.

\begin{figure}[t]
\centering
\includegraphics[width=0.5\textwidth]{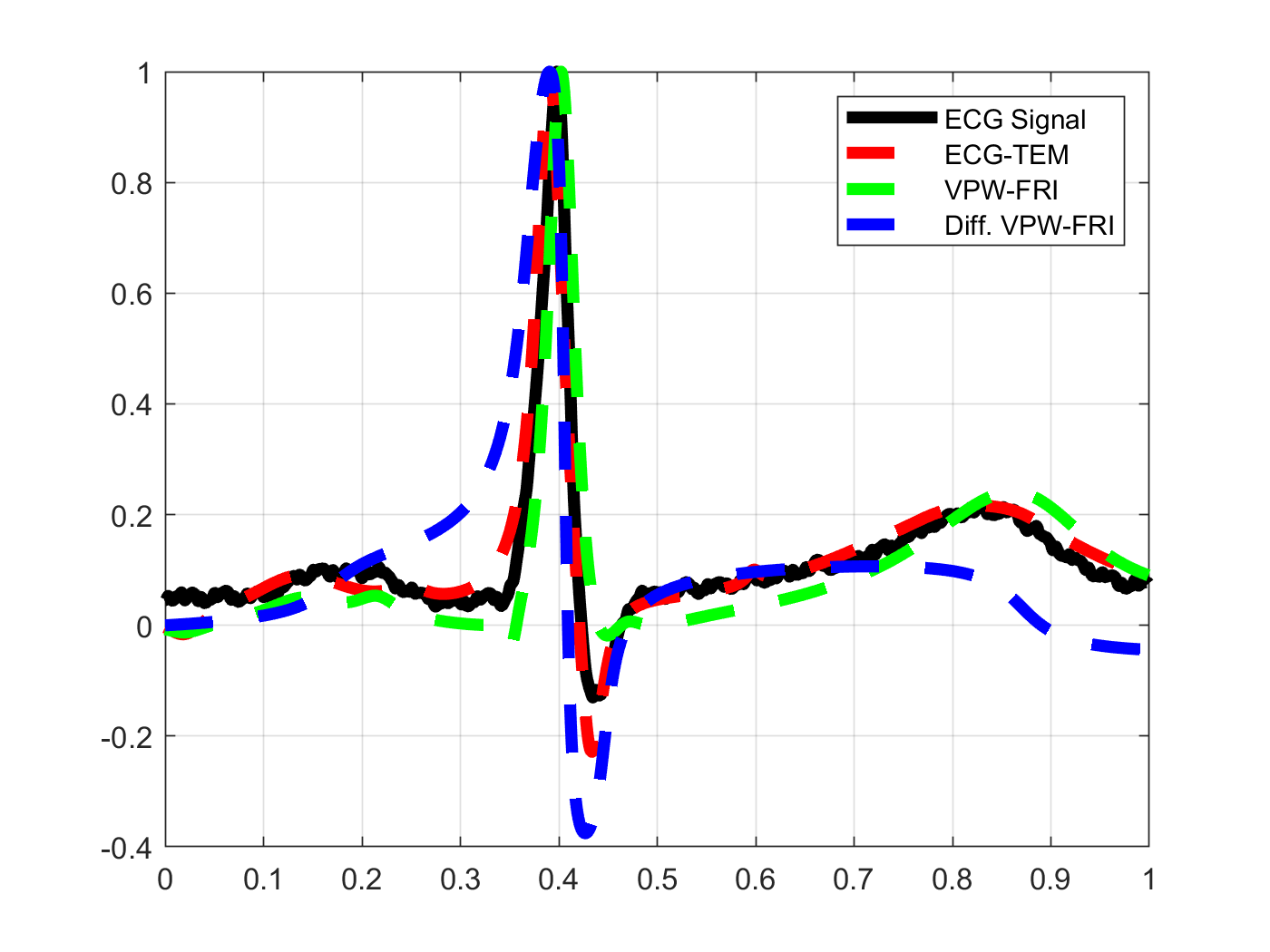}
\caption{ Reconstruction results of a single ECG pulse from \cite{schellenberger2020dataset}. Reconstruction is performed using 82 samples for each approach.}
\label{fig:ECG_All}
\end{figure}
\section{Simulations and HRM Application}
\label{sub:sim}
This section evaluates the performance of the proposed ECG-TEM approach for both ECG signal reconstruction and HRM. It compares it with existing uniform sampling techniques, namely the VPW-FRI method by \cite{baechler2017sampling} and the Differential VPW-FRI method by \cite{huang2022sub}. The evaluation is performed using real ECG recordings from 30 subjects obtained from the dataset provided by \cite{schellenberger2020dataset}.
Specifically, the resting scenario described in \cite{schellenberger2020dataset} is considered, where participants were lying on a table connected to several monitoring devices and instructed to breathe calmly and avoid large movements for at least 10 minutes.

The ECG signals are modeled as VPW-FRI pulses using the model in \eqref{eq:fri}, with $K=10$, as suggested by \cite{baechler2017sampling} and \cite{huang2022sub}. Fig. \ref{fig:ECG_All} shows an example of a single ECG pulse reconstruction, demonstrating the effectiveness of the proposed approach. As presented in Table 1, the ECG-TEM method provides the best reconstruction quality in terms of Root-Mean-Square Error (RMSE) compared to \cite{baechler2017sampling} and \cite{huang2022sub}.
Fig. \ref{fig:ECG_rec} illustrates the sampling and reconstruction process using the IF-TEM for an ECG signal. Subplot (a) depicts the sampling mechanism of a single filtered ECG pulse by the IF-TEM sampler, which generates a sequence of time instants representing the signal. Subplot (b) shows the reconstructed ECG signal obtained from the time instants produced by the IF-TEM sampler, demonstrating the capability of the proposed method to accurately recover the ECG signal.

\begin{figure}[!t]
\begin{center}
\begin{tabular}{cc}
\subfigure[]{\includegraphics[width = 3in]{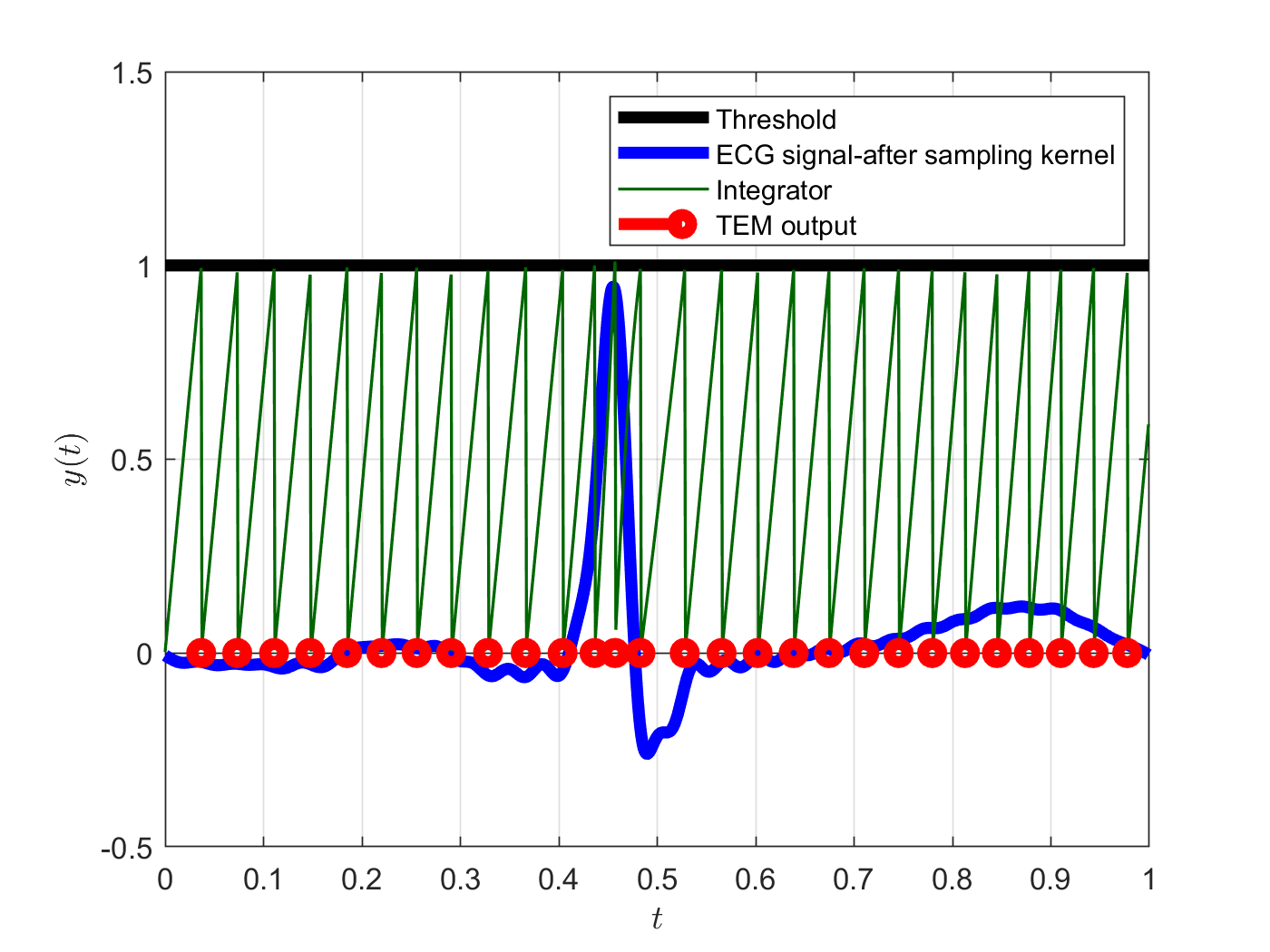}}\vspace{-.1in}\\
\subfigure[]{\includegraphics[width = 3in]{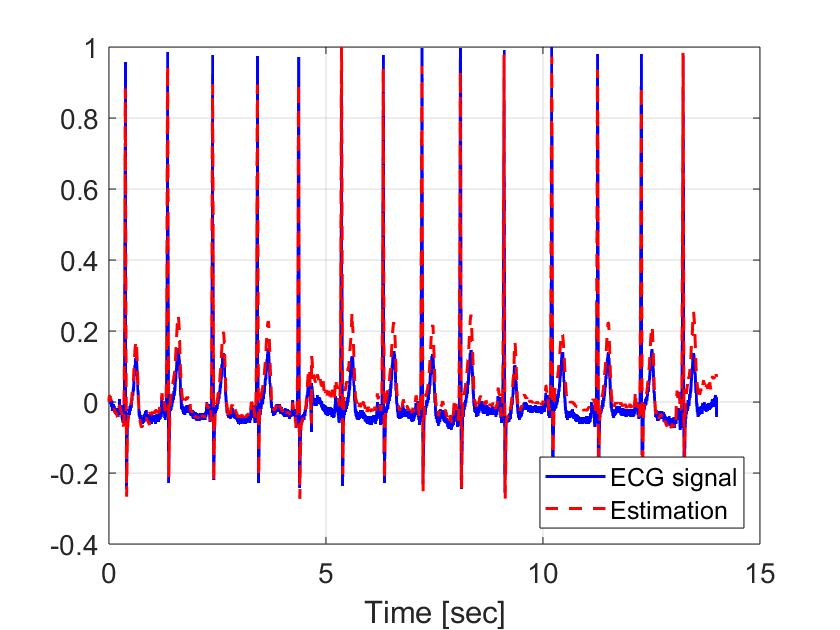}}
\end{tabular} 
\end{center}
\caption{IF-TEM sampling and reconstruction example. (a): Sampling mechanism of a single filtered ECG pulse by IF-TEM. (b): Reconstructed ECG signal. }
\label{fig:ECG_rec}
\end{figure}
\begin{figure}[ht]
\centering
\includegraphics[width=0.5\textwidth]{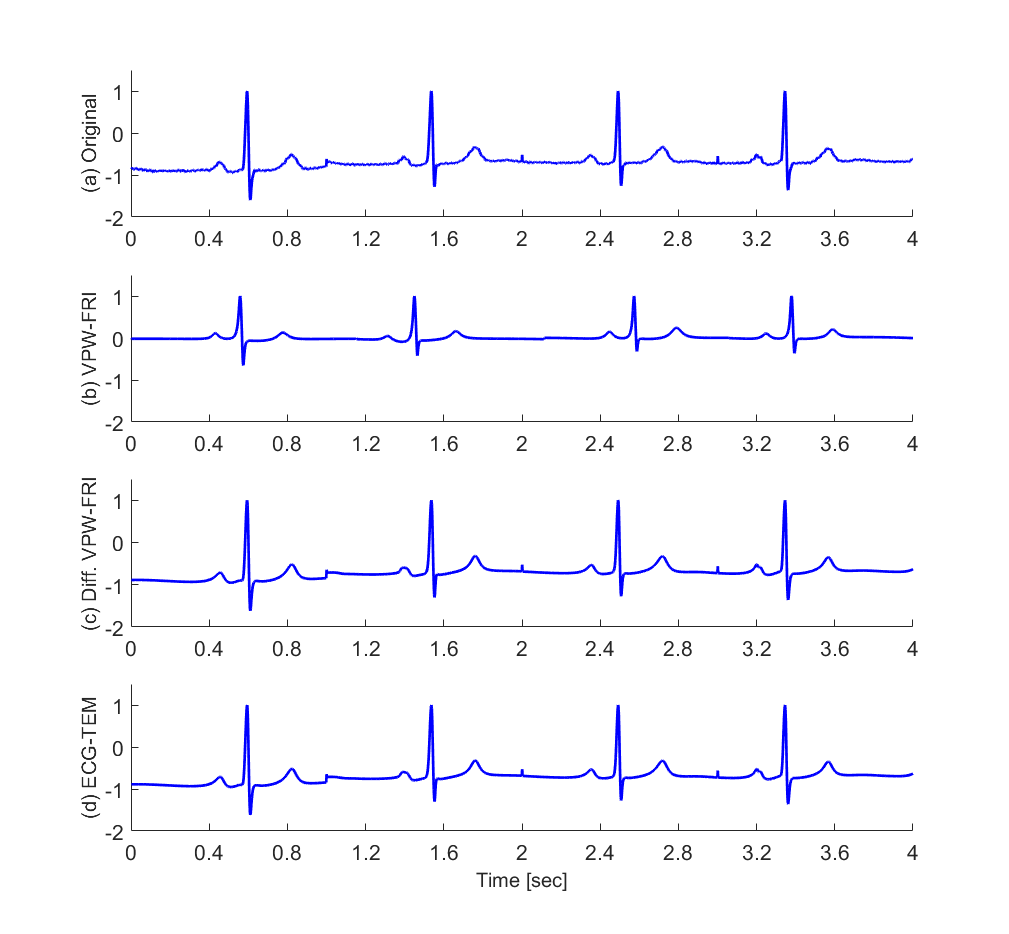}
\caption{ ECG signal reconstruction over a duration of four seconds: (a) Original ECG signal; (b) VPW-FRI reconstruction with 5 pulses per heartbeat (15 parameters/sec, $K=10$); (c) Differential VPW-FRI reconstruction with 5 pulses per heartbeat (15 parameters/sec, $K=10$); (d) ECG-TEM reconstruction with 5 pulses per heartbeat (15 parameters/sec, TEM parameters: $b=0.78$, $\delta=0.99$, and $\kappa = 0.018$, $K=10$).}
\label{fig:ECG_4SEC}
\end{figure}
\begin{table}
\caption{HR RMSE accuracy}
\vspace{-1em}
\begin{center}
\begin{tabular}{| l | l|}
  \hline
  \hspace{0.01cm} Method & RMSE \\ 
  \hline
  ECG-TEM & $0.005 $  \\ 
  \hline
  VPW-FRI \cite{baechler2017sampling} &$0.029$ \\
  \hline
      Differential VPW-FRI \cite{huang2022sub} &$0.14 $   \\
  \hline
\end{tabular}
\label{Table2}
\end{center}
\end{table}

Fig. \ref{fig:ECG_4SEC} depicts the ECG reconstruction for the first 4 seconds of the resting mode of record GDN0021 from the database \cite{schellenberger2020dataset}. Each patient in the dataset has 10 minutes of respiratory data, sampled at a standard 60 bpm with a 1Hz heart rate and 2000 samples per second. The selected signals were all sampled in the "rested" position.
Subplot (a) shows the original ECG signal, while subplots (b), (c), and (d) display the reconstructions obtained using different sampling and reconstruction methods:
(b) VPW-FRI estimation with 5 pulses per heartbeat, employing uniform sampling.
(c) Differentiated VPW-FRI estimation, also employing uniform sampling.
(d) ECG-TEM estimation, employing the proposed IF-TEM non-uniform sampling approach.
The parameter $K$, representing the number of pulses in the signal model, is set to 9 for all methods. The IF-TEM parameters used are $b=0.78$, $\delta=0.99$, and $\kappa = 0.018$.
Examining the results, it is evident that the ECG-TEM estimation in subplot (d) provides the most accurate reconstruction, closely matching the original ECG signal. The differentiated VPW-FRI estimation, shown in subplot (c), captures the overall shape reasonably well; however, it exhibits some inaccuracies in the amplitude, particularly in the R-peak region, which is crucial for HRM applications.

The reconstructed ECG signal is utilized to perform HRM, which is an application calculated from the R-peaks of the reconstructed ECG signal. The HRM results are compared to the HR calculated from a known synchronously sampled ECG signal at 2000 Hz \cite{schellenberger2020dataset}, which serves as the ground truth (GT) reference.
To conduct HRM, we sample the ECG signal using the ECG-TEM approach and then recover it. Subsequently, every half a second, an FFT-based peak selection \cite{adib2015smart,mercuri2019vital,sacco2020fmcw,alizadeh2019remote,antolinos2020cardiopulmonary} is applied to the R-peaks obtained from the last 40 seconds of the reconstructed signal, based on the resting heartbeat frequency range. With these settings for 10-minute monitoring, 450 HR data points are obtained for each participant, which are then compared to the GT. Fig. \ref{fig:HR} demonstrates that the proposed HRM method using IF-TEM exhibits a close resemblance to the reference ECG, outperforming other techniques.


To evaluate the quality of the estimates, the following statistical metrics are used: 1. Success rate, defined as the percentage of time when the HR estimate differs from the reference output by less than 2 beats per minute (b.p.m.), 2. Pearson Correlation Coefficient (PCC), 3. Mean-Absolute Error (MAE), and 4. Root Mean Square Error(RMSE).
Various signal-to-noise ratio (SNR) cases are investigated, where the SNR is defined as the inverse of an independent and identically distributed Gaussian noise variance. Each SNR case involves ECG data of 30 individuals from \cite{schellenberger2020dataset}. For each statistical metric and SNR value, the performance score is calculated as the median across all 30 participants.

Fig. \ref{fig:HR_comparison} shows the Success-Rate, PCC, MAE, and RMSE for HR estimation by all examined methods as a function of SNR. The HRM based on IF-TEM outperforms other compared methods in all four metrics for every SNR value. Detailed median accuracy scores for SNR = 2 dB in Table 2 demonstrate superior HRM results of IF-TEM.


Compared to \cite{baechler2017sampling}, note that while we use the same reconstruction method after calculating the FSCs, our ECG-TEM method achieves superior reconstruction quality. This improvement is attributed to two key factors: 1) the filter employed in IF-TEM excludes the zero frequency (DC term), and 2) the FSCs are calculated from the time instances using partial summation, enhancing the overall robustness of the method \cite{naaman2021fri}.
When the filter used in IF-TEM is applied to the VPW-FRI method \cite{baechler2017sampling}, the results do not improve. This observation indicates that the better performance of IF-TEM is not solely due to the filter itself. Rather, the combined effect of the filter excluding the DC term and the partial summation technique used in the reconstruction contributes to the superior reconstruction quality and overall performance achieved by the IF-TEM method compared to \cite{baechler2017sampling}.

\begin{figure}[!t]
\begin{center}
\begin{tabular}{cc}
\subfigure[VPW-FRI]{\includegraphics[width = 1.8in]{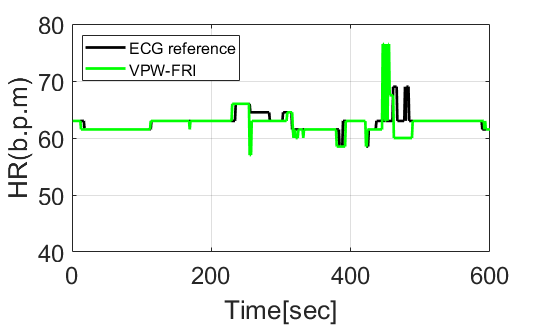}}&
\hspace{-.1in}\subfigure[Diff. VPW-FRI]{\includegraphics[width = 1.8in]{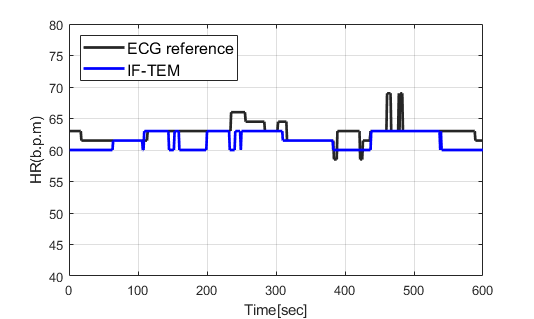}}\vspace{-.1in}
\end{tabular} 
\subfigure[ECG-TEM]{\includegraphics[width = 1.8in]{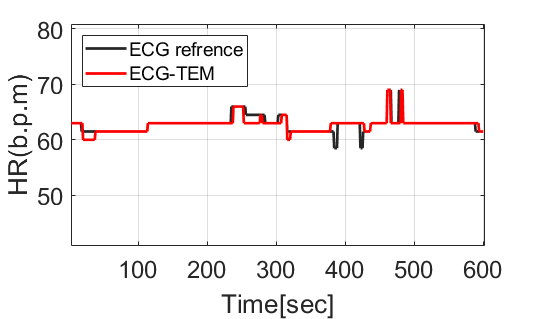}}\vspace{-.1in}
\end{center}
\caption{HR monitoring performance for $\textrm{SNR}=2$ [dB] (a): VPW-FRI \cite{baechler2017sampling}. (b): Differentiated VPW-FRI. (c) ECG-TEM. \cite{huang2022sub}.}
\label{fig:HR}
\end{figure}

\begin{figure}[!t]
\begin{center}
\begin{tabular}{cc}
\subfigure[]{\includegraphics[width = 1.7in]{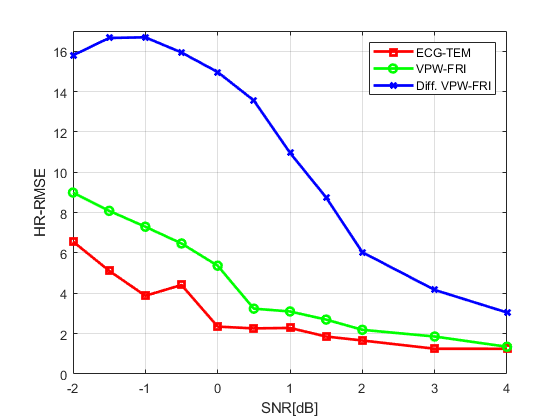}}&
\hspace{-.1in}\subfigure[]{\includegraphics[width = 1.7in]{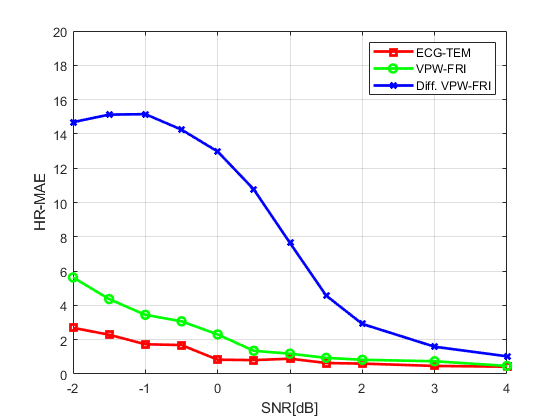}}\vspace{-.1in}\\
\subfigure[]{\includegraphics[width = 1.7in]{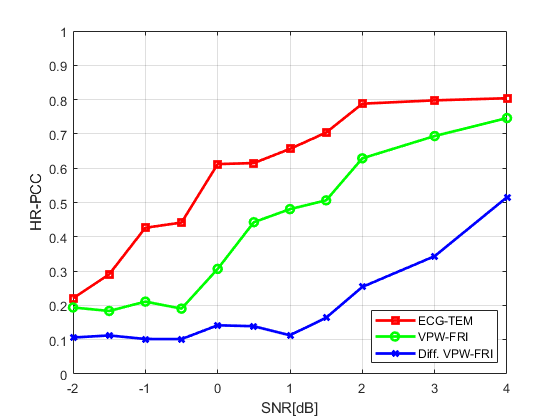}}&
\hspace{-.1in}\subfigure[]{\includegraphics[width = 1.7in]{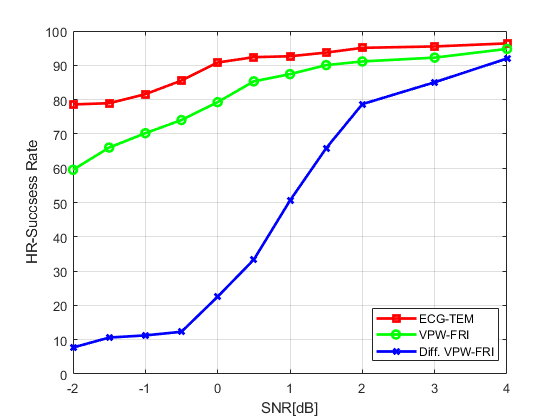}}\vspace{-.1in}
\end{tabular} 
\end{center}
\caption{ HRM performance vs. SNR. \textbf{(a)} HR RMSE. \textbf{(b)} HR MAE. \textbf{(c)} HR PCC. \textbf{(d)} HR Success Rate.}
\label{fig:HR_comparison}
\end{figure}


\section{IF-TEM Analog board and Hardware Experiments}
\label{sec:HW_spec}

\begin{figure*}[t!]
	\vspace{-0.5cm}
		\centering
		\includegraphics[width = 1\textwidth]{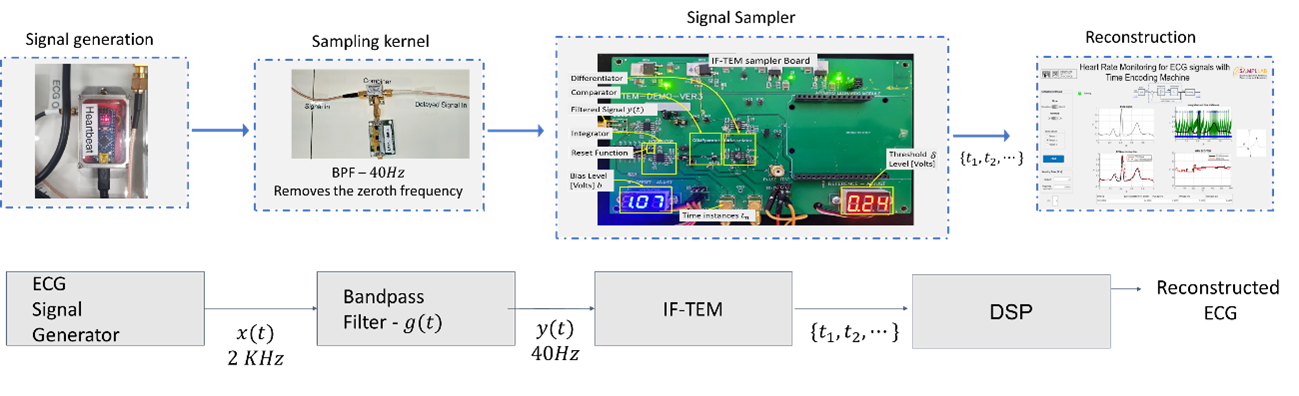}
		\vspace{-0.5cm}
		\caption{Block diagram of the ECG-TEM hardware prototype. The system consists of a signal generator, a sampling kernel (40 Hz BPF), and an IF-TEM sampler. The signal generator produces a realistic ECG signal, which is then filtered by the BPF to remove the zero-frequency component. The BPF in this setup is designed specifically for ECG signals, while the one in \cite{naaman2023hardware} was designed for FRI signals. The filtered signal is then sampled by the IF-TEM sampler board, which has been modified from the one in \cite{naaman2023hardware} by adjusting the RC values of the integrator component to accommodate ECG signal characteristics. The ECG signal, modeled as a VPW-FRI signal, is recovered from the sub-Nyquist samples obtained by the IF-TEM sampler using Algorithm 1.}
		\label{fig:hwflow}
\end{figure*}
\subsection{ECG-TEM Analog Board}

\begin{figure}[h!]
		\centering
		\includegraphics[width=0.5\textwidth]{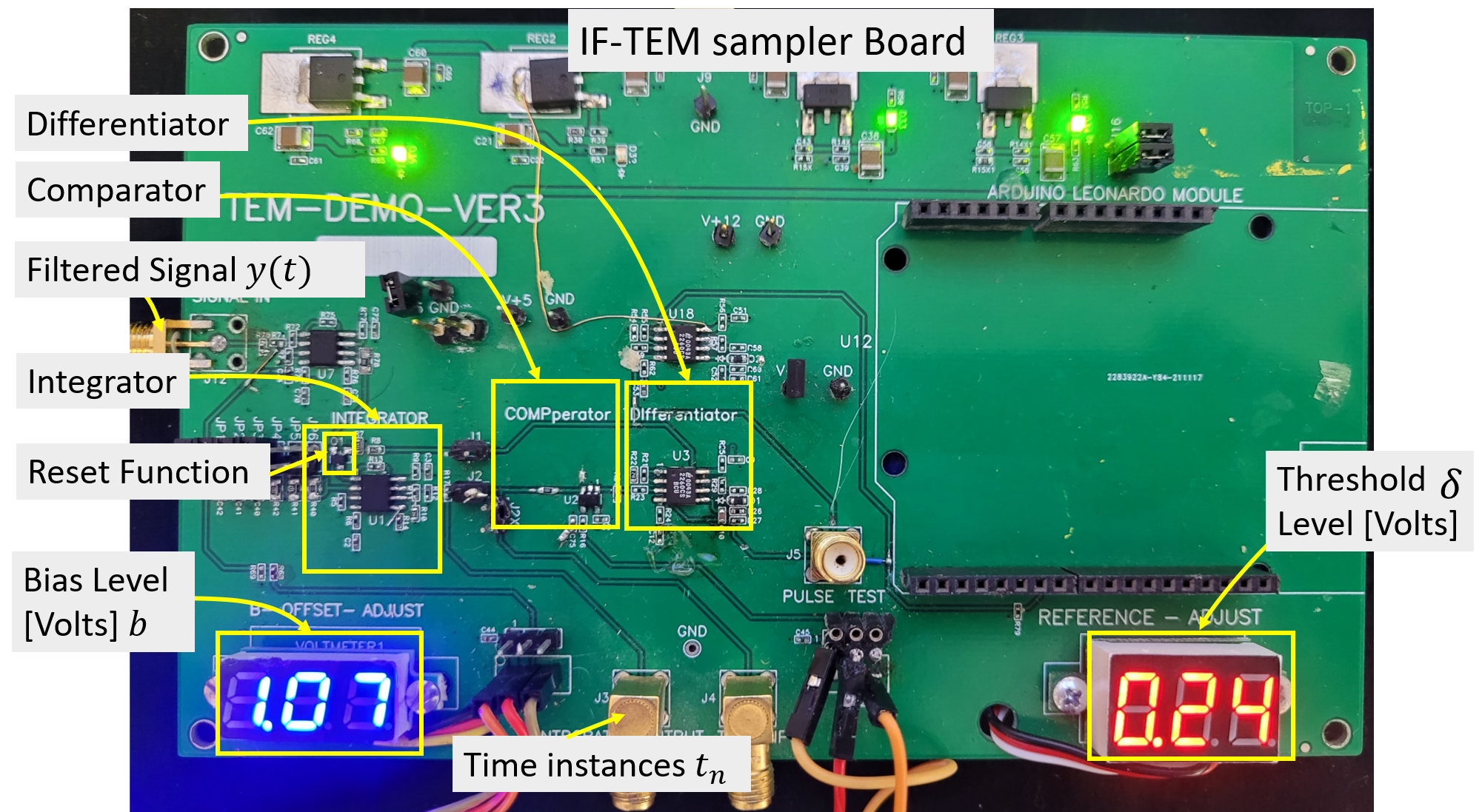}
		\caption{IF-TEM hardware board \cite{naaman2023hardware}.}
		\label{fig:IF-TEM_board}
\end{figure}
\begin{figure*}[t!]
    \centering
    \includegraphics[width=0.7\textwidth]{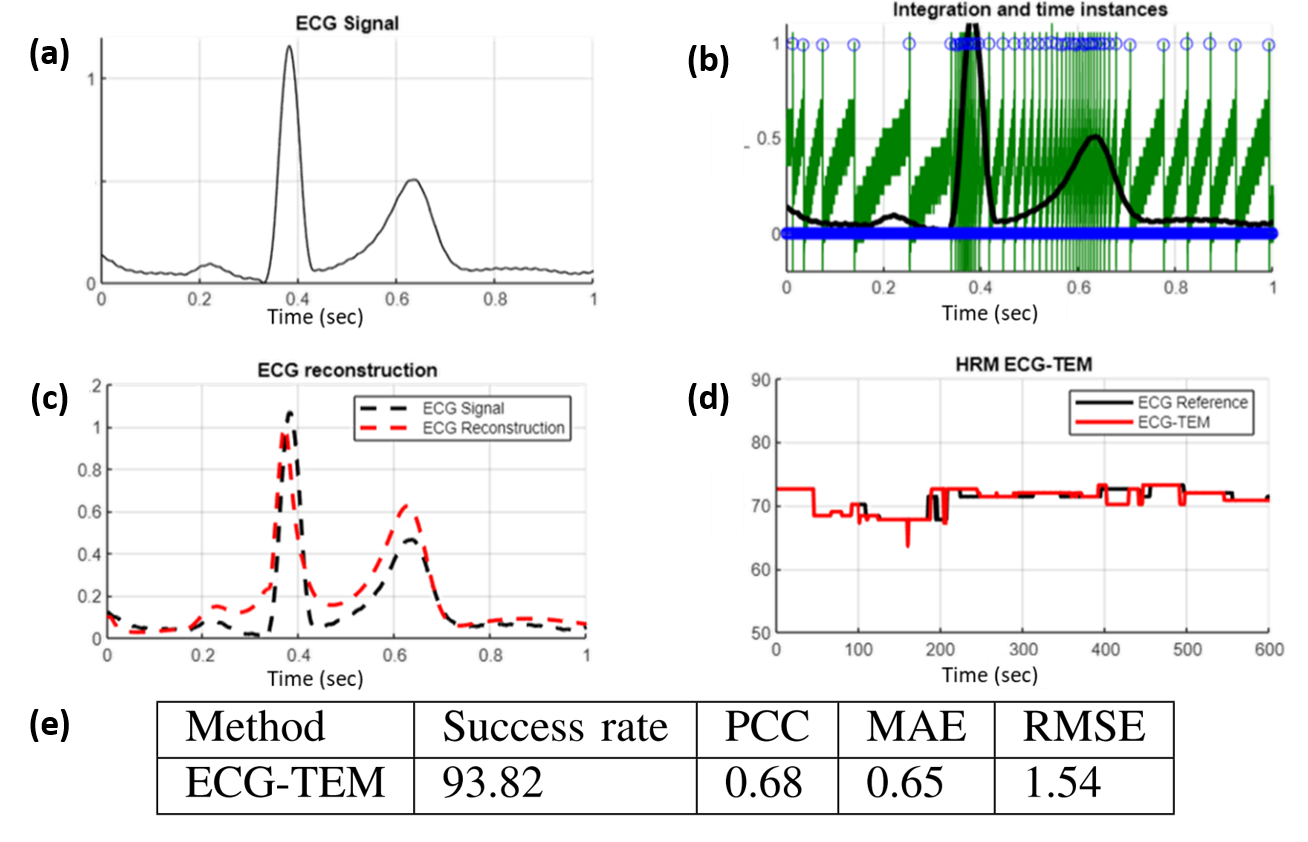}
    \caption{(a). ECG input signal $x(t)$ (b) BPF output $y(t)$ (green), and the IF-TEM output resulting in 49 samples (blue). (c). sampling and reconstruction using IF-TEM hardware: the input signal $x(t)$ (black) and its reconstruction (red). (d). HR estimate (red) with the reference output (black). (e). Different evaluation metrics used to assess the quality of our heart rate monitoring. }
    \label{fig:IAF0222}
\end{figure*}
\begin{figure*}[t!]
    \centering
    \includegraphics[width=0.7\textwidth]{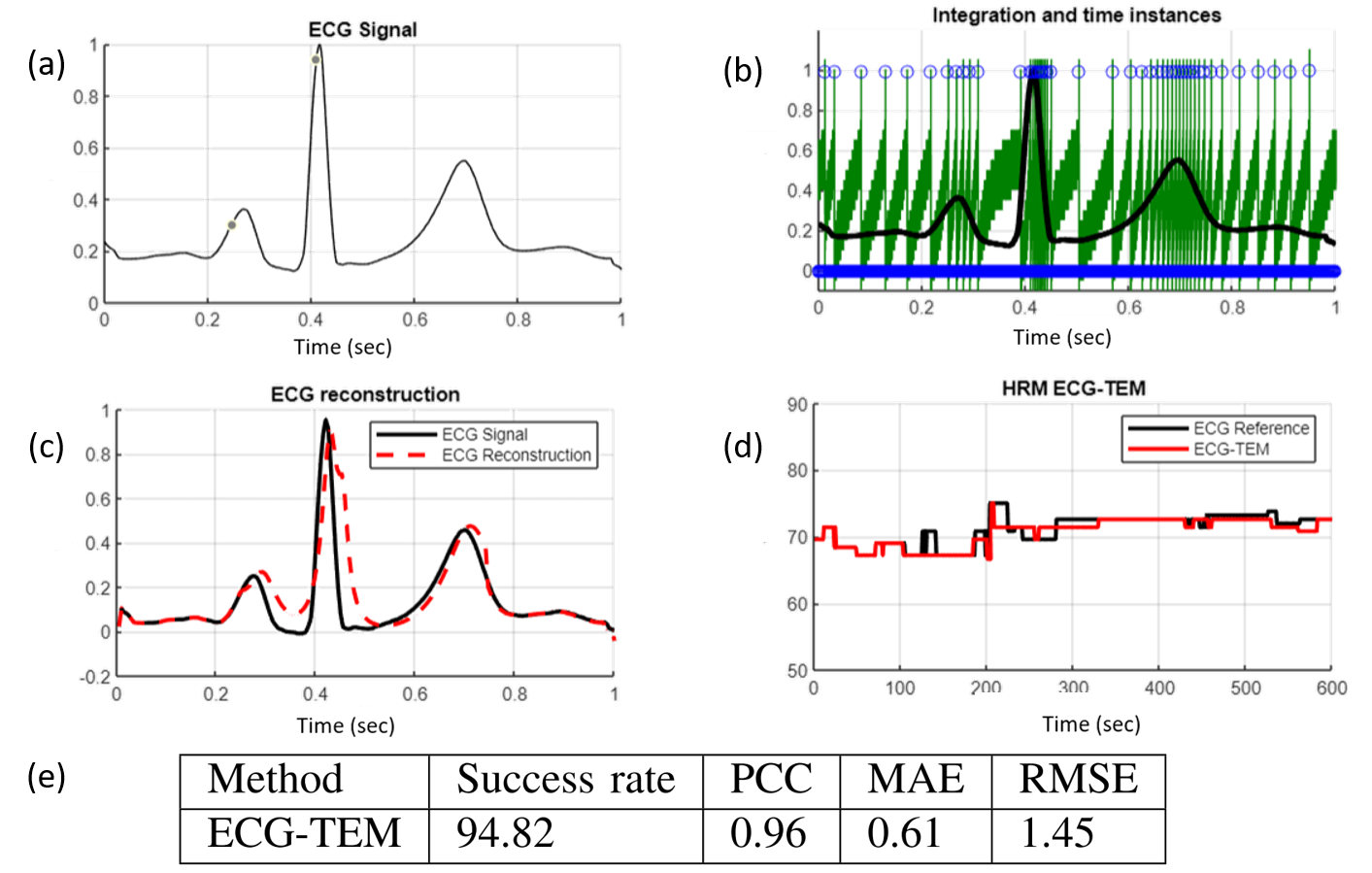}
    \caption{(a). ECG input signal $x(t)$ (b) BPF output $y(t)$ (green), and the IF-TEM output resulting in 46 samples (blue). (c). sampling and reconstruction using IF-TEM hardware: the input signal $x(t)$ (black) and its reconstruction (red). (d). HR estimate (red) with the reference output (black). (e). Different evaluation metrics used to assess the quality of our heart rate monitoring. }
    \label{fig:IAF0223}
\end{figure*}
This section describes the proposed ECG-TEM hardware prototype specifications. We begin by discussing the key components of the ECG-TEM hardware implementation, as well as various circuit design considerations. As shown in Fig. \ref{fig:hwflow}, the analog board comprises three sequential stages: the generation of an ECG signal, band-pass filtering, and an IF-TEM sampler.

In \cite{naaman2023hardware}, we introduced an IF-TEM hardware board designed for FRI signals. The IF-TEM circuit presented in this paper differs from the IF-TEM sampler in \cite{naaman2023hardware} by three main components: here, we designed a signal generator for injecting the ECG input, the filter is designed for ECG signals and not for FRI, and the RC values of the integrator component in the IF-TEM sampler are different. Specifically, for an ECG signal, a smaller capacitor and a larger resistor are required. Next, we specify the details of each component.

The ECG signal generator is a vital component responsible for generating the desired ECG waveform, which is a fundamental step in the overall functionality of the hardware system. The generated ECG signal possesses realistic amplitude and frequency characteristics, typically ranging from 0.05 Hz to 100 Hz. The filter also referred to as the sampling kernel, is employed to eliminate the zero-frequency component from the signal. The precise positioning of the sampling kernel, which essentially acts as a band-pass filter (BPF), is critical for the sub-Nyquist sampling and subsequent reconstruction of ECG signals using an IF-TEM (see Section \ref{sub:Kernel}). The output of the filter, denoted as $y(t)$ \eqref{eq:yt_by_x22_M}, is then fed into the IF-TEM sampler. A prototype of the IF-TEM sampler is illustrated in Fig. \ref{fig:IF-TEM_board} \cite{naaman2023hardware}.

The core components of the IF-TEM sampler comprise the bias $b$, an integrator, a comparator with threshold level $\delta$, and a reset function. To ensure sufficient samples for signal reconstruction, it is crucial to guarantee that the threshold $\delta$ is attained at least as many times as the desired number of samples. By introducing the bias $b$ to the input signal $y(t)$, the integrator receives a signal that is always non-negative. In this scenario, the integration of a non-negative signal results in a positive function, ensuring that the threshold is consistently reached. It is essential to carefully select an appropriate bias value to ensure the proper functioning of the IF-TEM system. The output of the integrator is then sent to the comparator, which compares the integrator voltage against a predefined threshold $\delta$. The threshold is a constant DC voltage that is implemented in hardware using a manually adjustable potentiometer.

The comparator is responsible for comparing the voltage generated by the integrator against a predefined threshold value. When the integrator voltage reaches or exceeds the threshold, the comparator's output changes state. If the comparator's input is below the threshold, it will output a logical '0', whereas if the input is above the threshold, the output will be '1'. In other words, the comparator will produce a sequence of logical '1' values whenever the integrator voltage hits the threshold. This change in the comparator's output signal indicates that the threshold has been reached and triggers the subsequent stage in the IF-TEM process.

The output of the comparator is then fed into a differentiator, which generates a short pulse that activates the fast reset function. The reset function consists of an amplifier and a field-effect transistor (FET) that work together to rapidly and completely discharge the integrator capacitor.

Next, we present the hardware experiments for ECG signal reconstruction and HRM.

\begin{table}
\begin{center}
\caption{HR estimation - median accuracy}
\vspace{-1em}
\begin{tabular}{| l | l|l|l|l|}
  \hline
  \hspace{0.01cm} Method & Success rate  & PCC & MAE & RMSE \\ 
  \hline
  ECG-TEM & $95.1 \%$ & $ 0.78$ & $0.60$  & $1.7$ \\ 
  \hline
  VPW-FRI \cite{baechler2017sampling} &$91.1 \%$ & $ 0.63$ & $0.83$  & $2.19$ \\
  \hline
      Differential VPW-FRI \cite{huang2022sub} &$78.7 \%$ &  $ 0.25 $ & $ 2.92  $ & $6.03$ \\
  \hline
\end{tabular}
\label{Table11}
\end{center}
\end{table}
\subsection{Hardware Experiments}
\label{sec:HW_experiments}

To evaluate the potential and feasibility of our system, we conducted experiments on the proposed ECG-TEM hardware system that we constructed. As illustrated in Figure \ref{fig:IAF0223}(a), we considered a real ECG input signal, denoted $x(t)$, consisting of five pulses with varying widths. The sampling kernel discussed in Section \ref{sub:Kernel} was employed in these experiments. The parameters for the IF-TEM circuit were set to a value of $\kappa=3\times 10^{-8}$, with a bias of $b=3V$ and a threshold of $\delta=1.5V$ and $K=10$.
The IF-TEM parameters were carefully selected to comply with the constraints outlined in \eqref{eq:sample_bound0}.

As demonstrated in Figures \ref{fig:IAF0222}(b) and \ref{fig:IAF0223}(b), the filtered signal $y(t)$ was fed into an IF-TEM sampler, which produced 49 and 46 time instances $t_n$, respectively, resulting in a firing rate of 49 Hz and 46 Hz, which is approximately 2.5 times the rate of innovation and 0.025 times the Nyquist rate.
Figures \ref{fig:IAF0222}(c) and \ref{fig:IAF0223}(c) illustrate a comparison between the original input signal and the estimated signal. This comparison demonstrates that the parameters of the ECG system can be robustly estimated while operating at a rate lower than the Nyquist rate.
Figures \ref{fig:IAF0222}(d) and \ref{fig:IAF0223}(d) show the calculation of the HRM, which is derived from the R peaks of the ECG signal.
Finally, Figures \ref{fig:IAF0222}(e) and \ref{fig:IAF0223}(e) present the different metrics that were used to evaluate the quality of our HRM estimates: 1. Success rate, defined here as the percentage of time in which the HR estimate differed from the reference output by less than 2 b.p.m., 2. Pearson Correlation Coefficient (PCC), 3.MAE, and 4. RMSE.

Our TEM hardware enables efficient sub-Nyquist sampling and recovery of ECG signals, benefiting heart rate monitoring applications. The ECG signal is filtered to remove its zero-frequency component, improving noise resilience.

As we demonstrate with one reconstruction example for proof of concept, the processed filtered signal, $y(t)$, is sampled using an IF-TEM sampler, resulting in a firing rate of 42-80Hz, equivalent to approximately $1/20$-$1/40$ of the Nyquist rate. While we can share more examples, this single example illustrates that we can robustly recover the ECG signal using an IF-TEM sampler.

\section{Conclusion}
\label{sec:Conclusions}
In this paper, we studied sampling and reconstruction frameworks for ECG signals utilizing IF-TEMs. We provided theoretical guarantees for reconstruction of ECG signals modeled as VPW-FRI signals and presented an approach to portable HRM based on the TEM-ECG framework. The proposed method employs the IF-TEM as a power-efficient, time-based, asynchronous sampler, addressing the critical need for low power consumption and sampling rate in long-term ECG monitoring. Numerical experiments validate the effectiveness of our approach, demonstrating accurate reconstruction of ECG signals and heart rate estimation. Furthermore, hardware validations bridge the gap between theory and practical implementation, showcasing the potential for real-world applications.

\appendix
\label{app:a}
 The filtered signal $y(t)$ is defined in \eqref{eq:yt_by_x22_M} as
\begin{equation*}
      y(t) =\sum_{m\in \mathcal{M}} X[m] e^{jm\omega_0 t},
\end{equation*}
where $\omega_0 = \frac{2\pi}{T}$, $m\in \mathcal{M}=\{-M,\cdots,-1,1,\cdots,M \}$ with finite $M$ and $X[m]$ denotes the FSCs of $x(t)$. From \eqref{eq:fri}, 
\begin{equation}
    X[m] = \sum_{k=0}^{K-1} X_k[m],
\end{equation}
where $X_k[m]$ are the Fourier coefficients of $x_k(t)$ which are $T$-periodic.
Substituting this representation into the expression for $y(t)$ from \eqref{eq:yt_by_x22_M}, we get:

\begin{equation}
    y(t) = \sum_{k=0}^{K-1} \sum_{m\in \mathcal{M}} X_k[m] e^{jm\omega_0 t}.
\end{equation}

Given that $\mathcal{M}$ is a finite set, $k \leq K$, and $X_k[m]$ are the FSCs of $x_k(t)$, we conclude that the sum over $k$ and $m$ of $X_k[m]$ is real and finite.
Consider the set $\{X_k[m] e^{jm\omega_0 t}\}_{m\in\mathcal{M}}$, all of its elements are finite since they are composed of the finite FSCs $X_k[m]$ multiplied by the complex exponential term $e^{jm\omega_0 t}$. Moreover, the set is finite since $\mathcal{M}$ is finite.
Let $M_{X,k}$ denote the maximum magnitude of the set $\{X_k[m] e^{jm\omega_0 t}\}_{m\in\mathcal{M}}$. 
Thus, we have:
\begin{equation}
|X_k[m] e^{jm\omega_0 t}| \leq M_{X,k}, \quad \forall m \in \mathcal{M}.
\end{equation}
Now, considering the sum over $k$ and $m$, we have:
\begin{equation}
\label{eq:app}
|y(t)| \leq \sum_{k=0}^{K-1} \sum_{m\in \mathcal{M}} |X_k[m] e^{jm\omega_0 t}| \leq \sum_{k=0}^{K-1} \sum_{m\in \mathcal{M}} M_{X,k}.
\end{equation}
Since both $K$ and $\mathcal{M}$ are finite sets, the bound in \eqref{eq:app} is finite.


\bibliographystyle{ieeetr}
\bibliography{refs}
\end{document}